\documentclass[10pt]{article}%
\usepackage{amsfonts}
\usepackage{amsmath}
\usepackage{amssymb}
\usepackage{graphicx}
\usepackage[square,numbers]{natbib}
\usepackage[left=2.5cm,top=3.0cm,right=2.5cm,bottom=2.5 cm]{geometry}%
\usepackage{hyperref}
\usepackage{authblk}
\usepackage{comment}
\providecommand{\U}[1]{\protect\rule{.1in}{.1in}}
\newtheorem{theorem}{Theorem}

\newtheorem{lemma}[theorem]{Lemma}

\newtheorem{remark}[theorem]{Remark}

\newenvironment{proof}[1][Proof]{\noindent\textbf{#1.} }{\ \rule{0.5em}{0.5em}}
\newcommand{\ii}{\mathrm{i}}
\newcommand{\R}{{\mathbb R}}
\newcommand{\cR}{{\mathcal{R}}}

\renewcommand{\Im}{\mbox{\rm Im}\,}
\begin{document}
\title{\textbf{Revisiting the problem of existence of surface Rayleigh waves with  impedance  boundary conditions}}
%\thanksmarkseries{alph}
\author[,1]{Fabio A. Vallejo\thanks{Corresponding author.\newline  E-mail: fabioval@ciencias.unam.mx (F. Vallejo),\:\: https://orcid.org/0009-0004-2580-7414}}
\affil[1]{Instituto de Investigaciones en Matem\'aticas Aplicadas y en Sistemas. Universidad Nacional 
Aut\'onoma de M\'exico. Circuito Escolar s/n, Ciudad Universitaria, Ciudad de M\'{e}xico C.P. 04510 (M\'exico)}
\date{}
\maketitle
\begin{abstract}

This paper considers the problem of surface waves in an isotropic elastic half-space endowed with impedance boundary conditions as first  proposed by  Godoy et al. [Wave Motion 49 (2012), 585-594]. These conditions  are controlled by two impedance parameters, where the standard stress-free boundary condition is retrieved for zero impedance. While the existence of a unique surface wave (called Rayleigh wave) is well-established for the standard stress-free boundary condition, the introduction of more general boundary conditions may lead to the absence of surface waves or even cause the PDE boundary value problem to  become ill-posed.  For the case of Godoy's impedance boundary conditions, the problem of existence and uniqueness of a surface wave of Rayleigh type  
 was investigated  by means of the complex function method based on Cauchy-type integrals. However, this method is quite cumbersome and hard to apply. In this work, we present an alternative method based on elementary tools from calculus to deal with the problem. We consider a  particular case  where  both impedance parameters are non-zero and  demonstrate  the existence and uniqueness of the surface wave for all material and boundary parameter values. Numerical examples are presented to illustrate the effect of the impedance parameter on the speed of the surface wave. \\

\textbf{keywords:} Rayleigh waves, impedance boundary conditions,  boundary value problems of PDE, surface waves, elastic half-space.

\end{abstract}

\section{Introduction}

The existence and uniqueness of a surface wave called Rayleigh wave is a well-known property for an isotropic elastic half-space endowed with the standard stress-free boundary condition. A lot of investigations on Rayleigh wave propagation in general anisotropic elastic half-spaces have been conducted, where the  stress-free boundary condition constitutes the main paradigm \cite{Ting2002}. However, non-standard  boundary conditions of Newman type have proven to be suitable for surface wave propagation modeling in several problems \cite{Godoy1,Masky1}. A very common type of them are the so-called impedance boundary conditions,  which  are defined by linear relations between the unknown function and its derivatives at the surface \cite{Godoy1}. In the framework of elasticity,   Tiersten \cite{Tier1} used    impedance boundary conditions  to simulate the effects of a thin layer  over an elastic half-space. In a very recent paper, C.Q. Ru. \cite{CQRu23} retrieved Tiersten's boundary conditions in the study of elastic metasurfaces, a relevant topic for seismic wave propagation. Malischewsky \cite{Masky1}  showed the potential of  impedance boundary conditions to model seismic wave propagation through discontinuities \cite{Godoy1}. Although useful to model a lot of problems, 
 a mathematical difficulty arises: the existence and uniqueness of a surface wave is given by the existence of a unique real root of the so-called secular equation.  This is a cumbersome algebraic  equation that determines the solvability of the associated boundary value problem of PDEs by solutions in form of surface waves \cite{Godoy1,Masky1}. %This  might explain the lack of general results about  existence of elastic surface  waves under non-standard boundary conditions. 
  For an abridged list of references on this topic, see \cite{Godoy1,HaRiv62,Rahman1995,Masky2000,Masky04,VinhOg2004,Rahman06,Li2006}.  In this work, we consider an isotropic homogeneous elastic half-space subjected to the impedance boundary conditions of Tiersten-type proposed by Godoy et al.  \cite{Godoy1}. They are defined in terms of two impedance parameters (tangential and normal) and are a natural generalization of the stress-free boundary condition in the non-dispersive regime.  
 The associated secular equation has the form of the secular equation in the stress-free case (zero impedance) plus two additional terms involving the impedance. A significant difficulty in the analysis arises from the term accounting for both  non-zero normal and tangential impedance. In order to avoid this problematic term and simplify the analysis,  earlier works  addressed  the cases where one of the impedance parameters is zero, namely the tangential case \cite{Godoy1,Vinh17} (see Eq. \eqref{imped121} with $Z_2=0$)  and the normal case \cite{Pham21} (see Eq. \eqref{imped121} with $Z_1=0$). Strangely enough, in the latter case,  it was found that surface Rayleigh waves cannot exist for certain values  of the normal impedance parameter ($Z_2$). In a very recent work \cite{Pham24}, Giang and Vinh  addressed the general case via the complex function method based on Cauchy-type integrals. The existence of the Rayleigh wave was established through an explicit formula of the speed of the wave.  However, the method involves advanced techniques from complex analysis, and as mentioned by the authors,   its implementation is not simple.  % As mention by  the authors in ,... the main challenge is the behavior of the Cauchy-type integrals at the end points of the integrations paths. Moreover, the resulting formulas are cumbersome.
  %Due to the challenging nature of the associated  secular equation, only the two cases where one of the  impedance parameter are zero have been addressed in the literature. Godoy et al. \cite{Godoy1}  analyzed the tangential case (see Eq. \eqref{imped121} with $Z_2=0$) and demonstrated the existence and uniqueness of the surface Rayleigh wave for each value of the tangential  impedance parameter $Z_1$. Giang and Vinh \cite{Pham21} implemented the complex function method to analyzed  the normal case (see Eq. \eqref{imped121} with $Z_1=0$). Strangely enough, there exist values of the normal impedance parameter $Z_2$ for which  surface waves are not possible. 
  In this paper we revisit the problem by addressing a particular case of Godoy's boundary conditions where both the tangential and normal impedance parameter are non-zero (see Eq. \eqref{imped121} with $Z_1=Z,\: Z_2=-Z$). Through the mathematical analysis of the secular equation via standard calculus techniques, we establish both  the existence and uniqueness of the surface Rayleigh wave  for all real value of the impedance parameter. The aim of this paper is therefore to present an alternative  simpler approach to the complex function method  to study the existence of surface Rayleigh wave under impedance boundary conditions. %, which relies on advanced complex analysis  and is challenging to implement. 
 This might be useful in more general scenarios, such as those involving anisotropy or higher dimensions, where  the intricate algebraic form of secular equation prevents  establishing general results regarding the existence of surface Rayleigh waves in terms of the material and impedance parameters  \cite{CQRu23,VinhThi14,KaurSingh21,KaurSingh21a}. %The analysis is thus restricted to particular cases that often require numerical approaches. 

The paper is organized as follows.  Section \S \ref{secular1} presents a detailed  derivation of the associated secular equation. In Section \S \ref{rayanaly} we prove the existence and uniqueness of the surface wave for each value of the impedance parameter. Section \S \ref{numeric} describes the effect of the impedance parameter on the speed of the surface wave. Numerical examples are also presented to illustrate the established behavior for three elastic materials.

\section{\textbf{Secular equation}}

\label{secular1}

In this section, we present briefly the derivation of secular equation of surface Rayleigh waves propagating in an elastic isotropic half-space subjected to full-impedance boundary conditions. Additional details can be found in \cite{Godoy1,Masky1}. Let us consider an isotropic elastic half-space with constant mass density $\rho$ occupying the domain $\{x_2\geq0\}$. The elastic properties are determined by  the standard  Lam\'e constants $\mu,\lambda$, which satisfy $$\mu>0,\;\;\lambda+\mu>0.$$
 In terms of  the Young's modulus $E$ and the Poisson's ratio $\nu$ %$\lambda=\frac{E\nu}{(1+\nu)(1-2\nu)},\;\;\mu=\frac{E}{2(1+\nu)}.$ 
  the latter inequalities are a consequence of the  usual inequalities  $E>0$ and $-1<\nu<0.5$ (see \cite{Achen75}). We shall study planar motion in the $(x_1,x_2)$-plane, such that the components of the displacement satisfy
$u_j=u_j(x_1,x_2,t),\;\text{for}\;j=1,2,\;\text{and}\;u_3\equiv 0$. The constitutive isotropic equations characterized by  the symmetric stress tensor $\sigma$  has four relevant components related to the displacement gradients by
\begin{equation}\label{streslin}
\sigma_{11}=(\lambda+2\mu)u_{1,1}+\lambda u_{2,2},\quad
\sigma_{12}=\sigma_{21}=\mu\big(u_{1,2}+u_{2,1}\big),\quad
\sigma_{22}=(\lambda+2\mu)u_{2,2}+\lambda u_{1,1},
\end{equation}
where commas denotes differentiation with respect to spatial variables $x_i$. In absence of source terms, the equations of motion are given by:
\begin{equation}\label{stssq1}
\sigma_{11,1}+\sigma_{12,2}=\rho\ddot{u}_1,\quad
\sigma_{12,1}+\sigma_{22,2}=\rho\ddot{u}_2.
\end{equation}

Assume that the components of the displacement vector depend harmonically on time through $e^{-\ii\omega t}$. That is, $u_j=e^{-\ii\omega t}\hat{u}_j(x_1,x_2)$, $j=1,2$, where $\omega$ is the angular frequency (cf. \cite{Godoy1}). Under such condition, we assume that the surface is subjected to the impedance boundary conditions   of the form (see \cite{Godoy1}) 
\begin{equation}\label{imped121}
\hat{\sigma}_{12}+\omega Z_1\hat{u}_1=0,\quad
\hat{\sigma}_{22}+\omega Z_2\hat{u}_2=0,
 \:\:\: \text{at}\:\:x_2=0,
\end{equation}
where $\hat{\sigma}_{jl}$ denotes  the Fourier components of the stresses $\sigma_{jl}$ ($j,l=1,2$), namely $\sigma_{jl}=e^{-\ii\omega t}\hat{\sigma}_{jl}(x_1,x_2)$. $Z_1,Z_2$ are  the impedance parameters having dimension of stress/velocity \cite{Godoy1}. The standard stress-free boundary condition is retrieved for zero impedance  $Z_1=Z_2=0$. When $Z_2=0$, we retrieve the tangential boundary condition investigated in \cite{Godoy1,Vinh17} and when only $Z_1=0$ we obtain the normal impedance boundary condition studied  in \cite{Pham21}. In this paper we address the case $Z_1=Z$, $Z_2=-Z$, $Z\in\R$.

To derive the secular equation, we proceed as in \cite{Pham21,Vinh17}, considering a surface wave of Rayleigh type that  propagate in the $x_1$ direction with velocity $c$, wave number $k>0$ and angular frequency $\omega=kc$. Its displacement components has the form
\begin{equation}\label{helch99}
\begin{split}
u_1(x_1,x_2)&=A\textit{\huge e}^{-bk x_2}\textit{\huge e}^{k\ii (x_1-ct)},\\
u_2(x_1,x_2)&=B\textit{\huge e}^{-bk x_2}\textit{\huge e}^{k\ii (x_1-ct)},
\end{split}
\end{equation}
%{\huge e}^{k\ii(x_1-ct)}
where $b$ (with $ b>0$), $A, B$, and  $c$ are unknowns to be determinated. The term ``Rayleigh type"  means that the displacement components satisfy the decay condition
\begin{equation}\label{deca3}
u_j=0,\;\text{as}\; x_2\to +\infty,\; j=1,2.
\end{equation}
We look for solutions to the boundary value problem \eqref{stssq1}-\eqref{imped121}  in form of  surface waves of Rayleigh type \eqref{helch99}. %First, we substitute \eqref{helch99} into the differential equation \eqref{stssq1}. 
 To do this, we first compute the stress components associated to \eqref{helch99} using \eqref{streslin}. This yields
\begin{equation}\label{9helch}
\begin{split}
\sigma_{11}&=\big((\lambda+2\mu) k\ii A-\lambda b k B\big)\textit{\huge e}^{-b k x_2}\textit{\huge e}^{k\ii (x_1-ct)},\\
\sigma_{12}=\sigma_{21}&=\mu(-b k A+k\ii B)\textit{\huge e}^{-b k x_2}\textit{\huge e}^{k\ii (x_1-ct)},\\
\sigma_{22}&=\big(\lambda k\ii A-(\lambda+2\mu)b k B\big)\textit{\huge e}^{-b k x_2}\textit{\huge e}^{k\ii (x_1-ct)}.
\end{split}
\end{equation}
Direct substitution of \eqref{9helch} into the differential equation \eqref{stssq1} and straightforward algebraic simplifications lead to the following linear homogeneous system in the variable $(A,B)$  
\begin{equation}\label{eq9helch}
\begin{split}
\big((c^2-c_1^2)+b^2c_2^2\big)A-b\ii(c_1^2-c_2^2)B=0,\\
-b\ii(c_1^2-c_2^2)A+\big((c^2-c_2^2)+b^2c_1^2\big)B=0,
\end{split}
\end{equation}
where $c_1:=\sqrt{(\lambda+2\mu)/\rho},\;c_2:=\sqrt{\mu/\rho}$ are the speed of the pressure and shear wave, respectively. %inasmuch as $(\lambda+\mu)/\rho=c_1^2-c_2^2$. 
 Non-trivial solutions of the system above are necessary to have non-trivial solutions of the form \eqref{helch99}, so  the determinant of the system \eqref{eq9helch} must vanish. After algebraic manipulation, we find that this happens iff $b=\pm b_1$ or $b=\pm b_2$, where
 %\begin{equation}\label{heldet1}
%\big(b^2c_1^2-(c_1^2-c^2)\big)\big(b^2c_2^2-(c_2^2-c^2)\big)=0
%\end{equation}
\begin{equation}\label{const15}
b_1=\sqrt{1-\dfrac{c^2}{c_1^2}},\;\;\;\;\; b_2=\sqrt{1-\dfrac{c^2}{c_2^2}}.
\end{equation}
Note that $b_1$ and $b_2$ are positive real number, provided that $0<c<c_2<c_1$ \cite{Pham21}. In order to fulfill the decaying condition outlined in \eqref{deca3}, we select the positive solutions for $b$, namely $b=b_j$, $j=1,2$. Now, solving the system \eqref{eq9helch}  for each value $b=b_1$, $b=b_2$ (separately) shows that the infinite set of solutions, for each value $b_j$, are respectively spanned by the vectors
\begin{equation}\label{solalge1}
\begin{pmatrix}A\\  B\end{pmatrix}=\begin{pmatrix}1\\    \frac{-b_1}{\ii}\end{pmatrix}\qquad \qquad \begin{pmatrix}A\\  B\end{pmatrix}=\begin{pmatrix}1\\    \frac{\ii}{b_2}\end{pmatrix}.
\end{equation}
Replacing in \eqref{helch99}, we obtain two linear independent Rayleigh-type  solutions $u^{(1)}$, $u^{(2)}$  with  components  given by
\begin{equation}\label{Raymod3q}
\begin{split}
u^{(1)}_1&=\textit{\huge e}^{-b_1 k x_2}\textit{\huge e}^{k\ii (x_1-ct)}\\
u^{(1)}_2&=\frac{-b_1}{\ii}\textit{\huge e}^{-b_1 k x_2}\textit{\huge e}^{k\ii (x_1-ct)}
\end{split}\;\;\;\;\;
\begin{split}
u^{(2)}_1&=\textit{\huge e}^{-b_2 k x_2}\textit{\huge e}^{k\ii (x_1-ct)}\\
u^{(2)}_2&=\frac{\ii}{b_2}\textit{\huge e}^{-b_2 k x_2}\textit{\huge e}^{k\ii (x_1-ct)}
\end{split}
\end{equation}
%\begin{equation}\label{Raymod3q}
%\begin{split}
%u^{(1)}_1&=\textit{\huge e}^{-b_1 k x_2}\textit{\huge e}^{k\ii(x_1-ct)}\\
%u^{(1)}_2&=\frac{-b_1}{\ii}\textit{\huge e}^{-b_1 k x_2}\textit{\huge e}^{k\ii(x_1-ct)}
%\end{split}\;\;\;\;\;
%\begin{split}
%u^{(2)}_1&=\textit{\huge e}^{-b_2 k x_2}\textit{\huge e}^{k\ii(x_1-ct)}\\
%u^{(2)}_2&=\frac{\ii}{b_2}\textit{\huge e}^{-b_2 k x_2}\textit{\huge e}^{k\ii(x_1-ct)}
%\end{split}
%\end{equation}
They are associated to the $P$-wave and $S$-wave, respectively (bulk waves). If we consider just one of these solutions, then there are values of the impedance parameters $Z_1,Z_2$ for which the boundary condition does not hold. Indeed, let us take for instance any scalar multiple of $u^{(1)}$ in \eqref{Raymod3q}. Since  the boundary condition \eqref{imped121} is expressed in terms of  $\hat{\sigma}_{jl}=\sigma_{jl} e^{\ii\omega t}$ ($j,l=1,2$)  with $\omega=k c$, we can easily retrieve $\hat{\sigma}_{12}$, $\hat{\sigma}_{22}$ (corresponding to $u^{(1)}$) from  the first solution in \eqref{solalge1} and using \eqref{9helch}.   Substituting into the boundary condition \eqref{imped121} yields the following algebraic system of equations
\begin{equation}\label{bounfake}
\left\{\begin{matrix}
-2\mu k b_1+Z_1\omega=0\\
\lambda k \ii-(\lambda+2\mu)b_1^2 k\ii+Z_2 \omega  b_1\ii=0.
\end{matrix}\right. 
\end{equation}
%\begin{equation}\label{bounfake}
%\left\{\begin{matrix}
%2\mu a_1+\gamma_1 kc\ii=0\\
%-\lambda k^2+(\lambda+2\mu)a_1^2+\gamma_2 a_1 kc\ii=0.
%\end{matrix}\right. 
%\end{equation}
Making $Z_1=0$, we find that the first equation holds only when $b_1=0$ (that is, $c=\pm c_1$) and hence, upon substitution of $b_1=0$, the second equation reduces to $\lambda=0$,  inasmuch as $k\neq0$. Consequently, once the Lam\'e constants are fixed so that $\lambda\neq0$, the first mode $u^{(1)}$ in  \eqref{Raymod3q} does not satisfy the boundary condition when $Z_1=0$ and $Z_2\in\R$.  Similarly, there are values of the impedance parameters for which the second mode $u^{(2)}$ in \eqref{Raymod3q} does not satisfy the boundary condition. Hence, a surface Rayleigh wave is necessarily a linear combination of $u^{(1)}$ and $u^{(2)}$ in \eqref{Raymod3q}. That is,  %me that  displacement components $u_j$, $j=1,2$ of  a surface wave of Rayleigh type satisfying the equation of motion \eqref{stssq1} and the decay condition \eqref{deca3} are of the form
  (see, \cite{Achen75,Pham21})
\begin{equation}\label{moddef56}
u_1=\big(A_1\textit{\huge e}^{-k b_1 x_2}+A_2\textit{\huge e}^{-k b_2 x_2}\big)\textit{\huge e}^{k\ii (x_1-ct)},\quad
u_2=\Big(-\frac{b_1}{\ii}A_1\textit{\huge e}^{-k b_1 x_2}+\frac{\ii}{b_2}A_2\textit{\huge e}^{-k b_2 x_2}\Big)\textit{\huge e}^{k\ii (x_1-ct)},
\end{equation}
where the constants  $A_j$ ($j=1,2$) have to be chosen so that  the impedance boundary condition \eqref{imped121} holds. %$b_1$, $b_2$ are given by ${b}_1=\sqrt{1-\tfrac{c^2}{c_1^2}},\;\;\;\;\;{b}_2=\sqrt{1-\tfrac{c^2}{c_2^2}},$. Moreover, $b_1$ and $b_2$ are positive real number, provided that $0<c<c_2<c_1$ \cite{Pham21}. This in turns ensures the fulfillment of the decay condition \eqref{deca3}.
 %The surface wave solution  \eqref{moddef56} must also satisfy the impedance boundary condition \eqref{imped121}. 
  On substituting the displacement $\hat{u}_j$ and the stresses $\hat{\sigma}_{jl}$ (by the usage of \eqref{9helch}) associated to \eqref{moddef56} into the boundary condition \eqref{imped121},  one obtains a homogeneous linear system of algebraic equations  for the amplitudes $A_1,A_2$ (see Equation (11) in \cite{Pham21})
 \begin{equation}\label{ampliec0}
\begin{aligned}
\big(-2\mu k b_1+\omega Z_1\big)A_1-\Big(\mu k(b_2+\frac{1}{b_2})-\omega Z_1\Big)A_2=0,\\
\Big(\frac{1}{\ii}(-\lambda k+(\lambda+2\mu)k b_1^2)+ \omega Z_2 b_1\ii\Big)A_1+\Big(-2\mu k\ii+\frac{\omega  Z_2\ii}{b_2}\Big)A_2=0.
\end{aligned}
\end{equation}
  Since we are looking for non-trivial surface waves ($A_1,A_2\neq(0,0)$), the determinant of the system hereabove must vanish. Simplifying and  taking into account that $c_2^2\rho(1+b_2^2)k=-\lambda k+(\lambda+2\mu)k b_1^2$, lead to the secular equation in the variable $c$
\begin{equation}\label{finsec5}
\begin{aligned}
\cR(c;Z_1,Z_2):&=\left(2-\dfrac{c^2}{c_2^2}\right)^2 -4\sqrt{1-\dfrac{c^2}{c_2^2}}\sqrt{1-\dfrac{c^2}{c_1^2}}+\dfrac{c^3}{\mu c_2^2}\left(Z_1\sqrt{1-\dfrac{c^2}{c_2^2}}+Z_2 \sqrt{1-\dfrac{ c^2}{c_1^2}}\right)\\ &+c^2\frac{Z_1 Z_2}{\mu^2}\left(\sqrt{1-\dfrac{c^2}{c_2^2}}\sqrt{1-\dfrac{c^2}{c_1^2}}-1\right)=0.
\end{aligned}
\end{equation}
The secular equation determines the solvability of the boundary value problem \eqref{stssq1}-\eqref{imped121}  by surface waves of Rayleigh type \eqref{helch99}. Thus, whenever a real   zero $c$ of \eqref{finsec5} occurs then there exists a surface Rayleigh wave. Note that  \eqref{goal1} has the form of the stress-free secular equation ($Z_1=Z_2=0$) plus two additional terms accounting for the impedance. In the tangential  and normal  impedance cases, ($Z_2=0$ and $Z_1=0$ respectively)  the associated secular equations involve only one additional term.
 \begin{remark}
 Recently,  the author proves in \cite{Fab23} that the secular equation  \eqref{finsecG} (in the variable $c$) does not possess roots outside the real axis for all $Z_1,Z_2\in\R$. This fact not only simplifies the problem of existence of a surface wave  but is also related to the well-posedness of the problem. Indeed, it is well-known in the  theory of hyperbolic PDEs, that solutions of the form \eqref{moddef56} associated to roots $c$ of the secular equation with $\Im c>0$ cause the ill-posedness of the associated boundary value problem of PDE. That is, the existence or uniqueness of the solution fail to hold once the data of the problem (L\'ame constants, source term, initial data) are prescribed. This is the case, for instance, when the boundary condition for system \eqref{stssq1} is defined by prescribing the stress vector to be proportional to the velocity at the boundary of the half-space (namely, $\sigma=\gamma\ \dot{u}$, $\gamma\in\R$).
 %This is the case, for instance, when the stress vector is set to be proportional to the velocity ($\sigma=\gamma\ \dot{u}$, $\gamma\in\R$) at the boundary of the half-space for the  system \eqref{stssq1}.  
  The associated secular equation  has at least one root with $\{\Im c>0\}$ for each positive value of the impedance parameter $\gamma$, implying the ill-posedness of the boundary value problem of PDE.  Conversely, when $\gamma<0$ the problem is well-posed  (refer to Proposition 5.1 in \cite{BRSZ} for details). Moreover, it was shown in \cite{Fab23} that negative values of $\gamma$ deviates the Rayleigh wave under the stress-free boundary condition (case $\gamma=0$) into a surface wave with a complex-valued velocity.%  implies the existence of a wave of infinite energy. Consequently, the problem is ill-posed under this range of parameters  (refer to Theorem 6.2 in \cite{Ser5} for details)

  %For instance when the boundary condition is fixed sich that the stress is proptional to the velocity with a real proportional ratio, the problem is ill-posed when the rato is positive ins as much as the secular equation has roots in the ipper complex half-plane. When the ratio is negative, the problem is well-posed in L2 with a surface wave that has a complex-valued velocity. A perturbed version of the Rayleigh wave (vanishing ratio) whose principal property is a complex-valued velocity.
 \end{remark}
%  Since we are looking for non-trivial surface waves, and thus, non-trivial solutions of that system  ($A_1,A_2\neq(0,0)$), the determinant of the associated matrix to the system must vanish. This leads to the secular equation for surface waves. 
 In terms of the dimensionless variables 
 \begin{equation}\label{newvar1}
 x=c^2/c_2^2,\;\:\delta_j=Z_j/ \sqrt{\mu\rho}\;(j=1,2)\;\:\text{and}\;\:\gamma=c_2^2/c_1^2=\frac{\mu}{\lambda+2\mu},
\end{equation}
the secular equation can be written as (see \cite{Pham21})
\begin{equation}\label{finsecG}
\begin{aligned}
f_1(\gamma,\delta_1,\delta_2,x):=(2-x)^2&-4\sqrt{1-x}\sqrt{1-\gamma x}+x\sqrt{x}\big(\delta_1\sqrt{1-x}+\delta_2\sqrt{1-\gamma x}\big)\\
&+\delta_1\delta_2x\big(\sqrt{1-x}\sqrt{1-\gamma x}-1\big)=0,
\end{aligned}
\end{equation}
Note that  the assumptions $\mu>0,\:\lambda+\mu>0$ are tantamount to $c_2<c_1$, so the material parameter $\gamma\in(0,1)$. Equation \eqref{finsecG} coincides with Eq. (12) in \cite{Pham21}. The case under consideration, namely $Z_1=Z$, $Z_2=-Z$, $Z\in\R$ amounts to setting the new variables as $\delta_1=\delta$, $\delta_2=-\delta$, $\delta\in\R$. Thus, the secular equation to analyze is given by
 \begin{equation}\label{goal1}
g(\gamma,\delta,x):=(2-x)^2-4\sqrt{1-x}\sqrt{1-\gamma x}+\delta x\sqrt{x}\big(\sqrt{1-x}-\sqrt{1-\gamma x}\big)\\
+\delta^2 x\big(1-\sqrt{1-x}\sqrt{1-\gamma x}\big)=0,
\end{equation}
We have the following criterion for the surface wave analysis to be carried out (see, Remark 1 in \cite{Vinh17})
\begin{remark}\label{anaray}
If a Rayleigh wave exists, then \eqref{goal1} has a solution $x_R$
so that $0 < x_R < 1$ and $x_R$ is the dimensionless squared velocity of the surface
 wave. Inversely, if \eqref{goal1} has a solution $x_R$ lying in the
interval $(0,1)$, then a surface wave is possible.
\end{remark}

\section{Surface  wave analysis}\label{rayanaly}
In this section we are going to prove the existence and uniqueness of the surface wave for the case in consideration by proving the existence of a unique root $x=x_R$ of the secular equation \eqref{goal1} in the interval $(0,1)$ for all $\gamma\in(0,1)$ and $\delta\in\R$. The existence is trivial. Indeed, observe that $g$ in \eqref{goal1} becomes a real-valued function of $x$ on the interval $[0,1]$ with $x=0$ a spurious root of $g=0$ with multiplicity $1$. Thus,  $f:=(1/x)g$ and $g$ have the same non-trivial roots on the interval $(0,1)$. So, our analysis is focused on the function  $f$ rather than $g$. This is given by
\begin{equation}\label{goal2}
f(x,\gamma,\delta):=\dfrac{(2-x)^2-4\sqrt{1-x}\sqrt{1-\gamma x}}{x}+\delta \sqrt{x}\big(\sqrt{1-x}-\sqrt{1-\gamma x}\big)
+\delta^2 \big(1-\sqrt{1-x}\sqrt{1-\gamma x}\big)=0.
\end{equation}
Note that $f$ can be defined continuously at $x=0$  inasmuch as
$\lim_{x\to0}f(x,\gamma,\delta)=-2(1-\gamma).$
After defining $f$ properly, it results continuous on $[0,1]$ with $f(0,\gamma,\delta)<0$, and  
$f(1,\gamma,\delta)=\big(\delta-\tfrac{1}{2}\sqrt{1-\gamma}\big)^2+\tfrac{1}{4}(\gamma+3)>0$
for all $\gamma\in (0,1)$ and $\delta\in\R$. The Intermediate value theorem ensures the existence of at least one root $x_R$ of $f$ on $(0,1)$ and thus a non trivial root of the secular equation \eqref{goal1} on $(0,1)$. We can  summarize as follows.
\begin{lemma}[Existence]
\label{maint1}
For $0<\gamma<1$ and the impedance parameter $\delta\in\R$, the secular equation \eqref{goal2} (or \eqref{goal1}) has at least one real solution $x_R$ within the range $(0,1)$.
\end{lemma}

The uniqueness of the surface Rayleigh wave is given by the uniqueness of the root in the lemma above. To prove this, we show that the first derivative $f'$ of  $x\to f(x;\gamma,\delta)$ is positive along $(0,1)$.  This implies the monotonicity of $f$ and then the uniqueness of the root.  
 However, the positiveness of $f'$ on $(0,1)$ shall be proved only for $0<\gamma<11/12$. When $\gamma\geq11/12$, the positive sign of the derivative holds along a sub-interval of the form $x\in(0,x_0)$, with $0<x_0<1$. In this case, the uniqueness of the root follows from the fact that $f$ is strictly positive along the rest of the interval, namely $[x_0,1)$. Thus, we shall consider the two cases separately. We use the following auxiliary  lemma for the proof.
 
 \begin{lemma}\label{adi23}
For all $\gamma\in (\tfrac{11}{12},1)$, it holds that
$\frac{4\sqrt{1-\gamma}}{\sqrt{1+3\gamma}}<\frac{1}{2\gamma}\Big(1-6\gamma+\sqrt{(1-6\gamma)^2+\gamma(7+16\gamma)}\Big)$
\end{lemma}
\begin{proof}
The result follows from the inequality  
\begin{equation}\label{ineq61}
a:=4\sqrt{1-\gamma}<1-6\gamma+\sqrt{(1-6\gamma)^2+\gamma(7+16\gamma)}=:b,\quad \gamma\in[11/12,1).
\end{equation}
Indeed, note first that $\sqrt{1+3\gamma}>2\gamma$ for all $\gamma\in(0,1)$. Thus, inverting the latter inequality and using \eqref{ineq61} give $a/\sqrt{1+3\gamma}<a/2\gamma<b/2\gamma$, which is the desired inequality. Let us prove \eqref{ineq61}. Define
$$h(\gamma):=1-6\gamma+\sqrt{(1-6\gamma)^2+\gamma(7+16\gamma)}-4\sqrt{1-\gamma}.$$
A straightforward calculation gives
$$h'(\gamma)=-6+\frac{2}{\sqrt{1-\gamma}}+\frac{7+32\gamma+12(6\gamma-1)}{\sqrt{(1-6\gamma)^2+\gamma(7+16\gamma)}}$$
Since $\gamma\in[11/12,1)$, it easy to verify that $6\gamma-1>0$ and also that  $-6+2/\sqrt{1-\gamma}>0$ (by removing the square root).
Thus $h'$ is positive on $[11/12,1)$. That is, $h$ is a strictly increasing function and since $h(11/12)=\frac{11}{12}-\frac{2}{\sqrt{3}}>0$, we conclude $h>0$ for all $11/12<\gamma<1$. This implies \eqref{ineq61} and the result.
\end{proof}

\begin{theorem}[Uniqueness]\label{uniroot}
The root in lemma \ref{maint1} is unique.
\end{theorem}
\begin{proof}
Let us define 
$\Theta:=\dfrac{\sqrt{1-\gamma x}}{\sqrt{1-x}}.$
Note that $\Theta>1$, provided that $x,\gamma\in (0,1)$. It is not hard to verify from \eqref{goal2} that the derivative of $x\to f(x,\gamma,\delta)$ is given by:
\begin{equation}\label{deriv0}
f'(x,\gamma,\delta)=A(x,\gamma)\delta^2+B(x,\gamma)\delta+C(x,\gamma),
\end{equation}
where the prime mark denotes the derivative with respect to $x$ and the coefficients are given by
\[
A(x,\gamma):=\Theta+\frac{\gamma}{\Theta},\quad
B(x,\gamma):=-\dfrac{1}{2\sqrt{x}}\left(1-\frac{1}{\Theta}\right)\left(2\Theta\sqrt{1-x}+\frac{1}{\sqrt{1- x}}\right),\quad
C(x,\gamma):=\dfrac{2\Theta}{x^2}\left(1-\frac{1}{\Theta}\right)^2+1.
\]
There are two cases to consider.\\

 \emph{\bf Case 1: $0<\gamma\leq11/12$.}
We shall prove that $f'$ is positive along $(0,1)$ for all $\delta\in\R$. We exploit the form of $f'$ as  a second-degree  polynomial in $\delta$. Note that  the coefficient of $\delta^2$ in \eqref{deriv0} is always positive. So, the positivity of $f'$ can be derived from the fact that there are no real roots of $f'$ in $\delta$ for each $x\in(0,1)$. That is, its algebraic discriminant has negative sign for $x\in(0,1)$ and $0<\gamma\leq11/12$. For the sake of simplicity, we prove the property for the polynomial $f'-1$, which also implies the property for $f'-1+1=f'$ by monotony. Since the coefficients of $f'-1$ are $A,B$ and $C-1$,  it is not hard to verify that the associated algebraic discriminant $\Delta$ is given by: 
\begin{equation}\label{discri1}
\Delta(x,\gamma) := \dfrac{1}{4x}\left(1-\frac{1}{\Theta}\right)^2\left(4(1-\gamma x)+\frac{1}{1-x}+4\Theta\right)- \dfrac{4}{x^2}\left(1-\frac{1}{\Theta}\right)^2(\Theta^2+\gamma).
\end{equation}
Since $x\in(0,1)$ and $\Theta>1$, then $1/x<1/x^2$ and $4\Theta<4\Theta^2$ (in the first term hereabove). Thus, we have
\begin{align}
\Delta(x,\gamma) &< \dfrac{1}{4x^2}\left(1-\frac{1}{\Theta}\right)^2\left(4(1-\gamma x)+\frac{1}{1-x}+4\Theta^2-16\Theta^2-16\gamma\right)\nonumber\\
&= \dfrac{1}{4x^2}\left(1-\frac{1}{\Theta}\right)^2\frac{1}{1-x}\Big(4\gamma x^2+4(6\gamma-1)x-(7+16\gamma)\Big).\label{discri2}
\end{align}
Since the remaining factors are positive, the sign of the last expression hereabove, and therefore the sign of $\Delta$, depends on the sign of the polynomial
$P_{\gamma}(x):=4\gamma x^2+4(6\gamma-1)x-(7+16\gamma).$
By the assumption $0<\gamma<11/12$, it is easy to verify that $P_{\gamma}$ takes negative values at the end-points of the interval $[0,1]$. Indeed,  $P_{\gamma}(0)=-(16\gamma+7)$ and $P_{\gamma}(1)=12\gamma-11\leq0$. Therefore, we claim $P_{\gamma}(x)\leq0$ for all $x\in(0,1)$. Otherwise, it would contradict the convexity of the quadratic function $x\to P_{\gamma}(x)$, provided the coefficient of $x^2$ is positive. Thus, the algebraic discriminant $\Delta(x,\gamma)$ has negative sign for all $x\in(0,1)$. This implies that $f'-1$, and thus also $f'$, does not have real roots in $\delta$. Since the coefficient $A$ of $\delta^2$  is always positive, then $f'$ (as a quadratic real-valued function of $\delta$) is positive on $\R$  for each $x\in(0,1)$ and $0<\gamma<11/19$. That is, $x\to f(x,\gamma,\delta)$ is strictly monotone on (0,1) for all $\delta\in\R$, implying the uniqueness of the root from lemma \ref{maint1}.\\

 \emph{\bf Case 2: $11/12<\gamma<1$.} 
 In this case, we can also use \eqref{discri2}. However, although  $P_{\gamma}(0)<0$ still remains valid, now $P_{\gamma}(1)=12\gamma-11>0$. Thus, in contrast to the latter case, there is a root $x=r_+$ of $P_{\gamma}$ on $(0,1)$. This root is
 $$r_+(\gamma):=\frac{1}{2\gamma}\Big(1-6\gamma+\sqrt{(1-6\gamma)^2+\gamma(7+16\gamma)}\Big).$$
 The another root is clearly negative. Therefore, in this case we have
 $P_{\gamma}(x)\leq0,\:\text{for all}\: x\in(0,r_+)$
 and then, \eqref{discri2} implies that $\Delta<0$ as long as $x\in(0,r_+)$. By the same argument as in the last case,  $x\to f(x,\gamma,\delta)$ is monotone along $x\in(0,r_+)$ for all $\delta\in\R$ and $11/12<\gamma<1$. We claim that the root from lemma \ref{maint1}  lies within  $(0,r_+)$ and therefore is unique, because as we shall see, $f$ is positive for all $x\in[r_+,1)$. Indeed, consider $f$ in \eqref{goal2} as a second-degree polynomial in $\delta$
\begin{equation}\label{deriv0}
f(x,\gamma)=C_0(x,\gamma)+B_0(x,\gamma)\delta+A_0(x,\gamma)\delta^2,
\end{equation}
where $A_0(x,\gamma):=1-\sqrt{1-x}\sqrt{1-\gamma x},\quad
B_0(x,\gamma):=\sqrt{x}\big(\sqrt{1-x}-\sqrt{1-\gamma x}\big),\quad
C_0(x,\gamma):=x-4+\frac{4}{x}A_0(x,\gamma).
$
We proceed as in the last case. Since $A_0$ is positive, it is enough to prove that there are no real roots of $f$ in $\delta$ for each $x\in[r_+,1)$. That is, the algebraic discriminant $\Delta_0$ has negative sign. A straightforward calculation gives
\begin{equation}\label{discri0}
\Delta_0(x,\gamma) := \dfrac{4-x}{x}\left(x^2+6x-8+(8-2x)\sqrt{1-x}\sqrt{1-\gamma x}\right)-\gamma(x^2+16x-16).
\end{equation}
Since $x,\gamma \in(0,1)$, then $8-2x>0$ and $\sqrt{1-x}<\sqrt{1-\gamma x}$. Thus, it follows that
\begin{align}
\Delta_0(x,\gamma) &< \dfrac{4-x}{x}\left(x^2+6x-8+(8-2x)(\sqrt{1-\gamma x})^2\right)-\gamma(x^2+16x-16)\nonumber\\
&= (1+3\gamma)\left(\frac{4\sqrt{1-\gamma}}{\sqrt{1+3\gamma}}+x\right)\left(\frac{4\sqrt{1-\gamma}}{\sqrt{1+3\gamma}}-x\right).\label{posi1}
\end{align}
Note that the sign of the last expression depends on the third factor, as the remaining ones are positive. It is not hard to verify that $0<\frac{4\sqrt{1-\gamma}}{\sqrt{1+3\gamma}}<1$ for all $\gamma>15/19$ and therefore, also for the range of $\gamma$ under consideration. Since $x\in(0,1)$, the second-degree polynomial in the right hand side of \eqref{posi1}, and thus $\Delta_0$, has negative sign for all $x\in(\frac{4\sqrt{1-\gamma}}{\sqrt{1+3\gamma}},1)$. Lemma \ref{adi23} ensures $\frac{4\sqrt{1-\gamma}}{\sqrt{1+3\gamma}}<r_+<1$. Therefore, in particular we have
$\Delta_0(x,\gamma)<0,\:\text{for all}\: x\in[r_+,1).$
This implies $f$ as a second-degree polynomial in $\delta$ does not have roots. Thus, given that the coefficient $A_0$ of $\delta^2$ is positive,  $f$ has positive sign for all $\delta\in\R$ and $x\in[r_+,1)$. 
\end{proof}

\begin{remark}

 Lemma \ref{maint1}  and Theorem \ref{uniroot} show that an isotropic elastic half-space subjected to  impedance boundary conditions of the form  \eqref{imped121} with $Z_1=Z$, $Z_2=-Z$, $Z\in\R$ supports a unique surface Rayleigh wave for all  $\gamma\in(0,1)$ (material parameter) and  $\delta=Z/ \sqrt{\mu\rho}\in\R$ (dimensionless impedance parameter). This result is consistent with that derived by Giang and Vinh  (see Theorem 1 (i) in \cite{Pham24}) via the complex function method. This is quite remarkable, considering that in the case of a normal impedance boundary condition with $Z_1 = 0$ and $Z_2 \in \mathbb{R}$, there are values of $Z_2$ for which surface waves cannot exist (see Theorem 3 in \cite{Pham21}). 

%The existence of  a unique surface wave called Rayleigh wave is a well-known property for the system of isotropic linear elasticity \eqref{} endowed with the stress free boundary condition \eqref{} with $Z=0$. Lemma  and Theorem extend this property to the impedance boundary condition \eqref{} with $Z_1=Z$, $Z_2=-Z$, $Z\in\R$. Note that the stress-free boundary condition is retrieved for zero impedance $Z=0$. This is quite remarkable, given that for the case of normal impedance boundary condition $Z_1=0, Z_2\in\R$ there are values $Z_2$ for which surface waves cannot exist. 
\end{remark}

\section{Numerical results}\label{numeric}
\begin{table}
\begin{center}
\vspace{.1cm}
\resizebox{!}{1.1cm}{\begin{tabular}{lcc}
\hline
Material &  $\nu$  & $\gamma$ \\ 
\hline
%\hline
Gold & 0.44 & 0.107 \\
%\hline
Chromium & 0.21 & 0.367 \\
%\hline
Polymer Foam & -0.7 & 0.705 \\
\hline
\end{tabular}}
\vspace{.1cm}
\caption{Poisson's ratio $\nu$ and the dimensionless parameter $\gamma$ of the materials considered.}
\label{materials}
\end{center}
%\vspace{-.1cm}
\end{table}

In this section, we describe how  the (squared) dimensionless surface wave speed $x_R$  is affected by the real  impedance $\delta\in\R$. Observe that for large values of $\delta$, the relevant term in the secular equation \eqref{goal2} is $A_0\delta^2=0$, which vanishes in $x\in[0,1]$ iff $x=0$. This indicates that the dimensionless surface wave speed $x_R$ approaches asymptotically to zero as the impedance goes to $\pm\infty$. To determine the behavior around $\delta=0$,   
 we can compute from \eqref{goal2}  the implicit derivative of  $x_R=x_R(\delta)$ for fixed material parameter $\gamma$. This is given by 
\begin{equation}\label{derivxr}
\frac{d\:x_R}{d \delta}=\dfrac{-B_0-2\delta A_0}{f'(x_R,\gamma,\delta)}.
\end{equation}
Recall that $f'$ is positive along $x\in(0,1)$ for all $\delta\in\R$, $\gamma\in(0,11/12)$. Thus, the sign of $\tfrac{d\:x_R}{d \delta}$  depends upon the numerator in $\eqref{derivxr}$. Since $B_0<0$ and $A_0>0$ for all $x,\gamma\in(0,1)$, it is easy to verify the positiveness of $\tfrac{d\:x_R}{d \delta}$  along $\delta\in(-\infty,0)$, including $\delta=0$. That is, $\delta\to x_R(\delta)$ is an increasing function along $(-\infty,0)$ and locally around a vicinity of $\delta=0$. Given the vanishing  behavior at infinity, $x_R$ attains at least a maximum for some positive value of the impedance. 
\begin{figure}[t]
\begin{center}
\includegraphics[scale=1.1, clip=true]{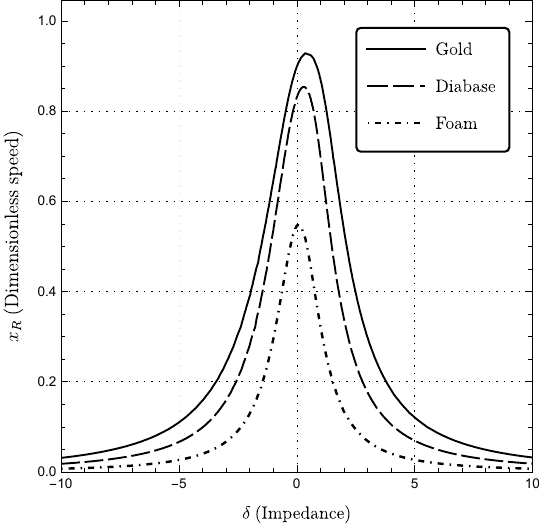} 
\end{center}
%\vspace{-.2cm}
\caption{Plot of the dimensionless surface wave speed $x_R$ vs the dimensionless impedance  $\delta$, for the three elastic material under consideration. 
 (Color online.)}\label{figCGa}
%\vspace{-.2cm}
\end{figure}
To  illustrate this behavior, we proceed as in \cite{Godoy1} by considering three elastic materials.   Namely, Gold (metal),  Chromium  and a polymer foam structure (auxetic material). A detailed description of these materials and their characteristics can be found in  \cite{lind68} (for Gold and Chromium) and in \cite{Lakes87} (for the foam structure). The role of the elastic properties of a particular material in the secular equation is given through the dimensionless parameter $\gamma$, which depends on the Lame's constants. For the sake of simplicity,  we find $\gamma=\mu/(\lambda+2\mu)$ (see Eq. \eqref{newvar1}) in terms of the Young's modulus $E$ and the Poisson's ratio $\nu$ through the usual relations (see \cite{Achen75}) $$\lambda=\frac{E\nu}{(1+\nu)(1-2\nu)},\;\;\mu=\frac{E}{2(1+\nu)}.$$ Straightforward calculations gives $\gamma=\tfrac{1-2\nu}{2(1-\nu)}.$
Thus, only the  Poisson's ratio $\nu$  is relevant for evaluating the role of each material  in the relation $\delta\to x_R$. Table \ref{materials} shows this elastic property for each considered material. 
To present the results, consider an impedance $\delta$ varying from a minimum value $\delta_{min} = -10$  to a maximum value $\delta_{max} = 10$. The speed $x_R$ is calculated by solving iteratively the secular equation \eqref{goal2}. The results are presented in Fig. \ref{figCGa}.\\

\section{Discussion}
In this work we investigated the existence of surface waves in an isotropic elastic half-space endowed with a uniparametric family of impedance boundary conditions with both non-zero tangential and normal impedance parameters. The existence and uniqueness of the surface wave for each value of the impedance parameter was demonstrated  by means of the mathematical analysis of the secular equation.  The speed of the surface wave decreases asymptotically to zero for large negative and positive values of the impedance, with a maximum value for some positive value of the impedance. The theoretical findings are verified by presenting numerical results for three elastic media. The general problem was investigated by Giang and Vinh \cite{Pham24} via the complex function method based on Cauchy-type integrals. We recover the same result obtained in that work for the case analyzed here, but using a simpler approach based on standard calculus techniques.  The method relies on the fact that, as a function of the impedance $\delta$,  the secular equation takes the form of a second-degree polynomial. 

In a very recent paper, Singh and Kaur \cite{KaurSingh21a} compute the secular equation for surface waves in an  incompressible micropolar medium subjected to the impedance boundary conditions here considered \eqref{imped121}. However, given the intricate final expression, the authors used a numerical approach to study  the effect of the impedance on the Rayleigh wave speed for particular cases. Notably,  when the impedance parameters have the same magnitude (either with equal or opposite sign), the secular equation also takes the form of a second-degree polynomial. This suggests that the method here presented could be useful for obtaining general results on the existence of surface Rayleigh waves in terms of the impedance for an  incompressible micropolar media endowed with impedance boundary conditions. 
\\

\textbf{Acknowledgements}:    We highly acknowledge the support of the National Science and Technology Council (CONAHCyT) of M\'exico for their funding of this research. This work was supported  under grant CF-2019 No. 304005 and partially under grant  CF-2023-G-122.\\\\
%This work was supported by the National Science and Technology Council (CONACyT) of M\'exico under grant CF-2019 No. 304005.

\textbf{Declarations}\\\\
\textbf{Conflict of interest} The author declares that he has no conflict of interest.

\bibliographystyle{abbrvnat}
\bibliography{bibliography}

\def\cprime{$'\!\!$}
\begin{thebibliography}{21}
\providecommand{\natexlab}[1]{#1}
\providecommand{\url}[1]{\texttt{#1}}
\expandafter\ifx\csname urlstyle\endcsname\relax
  \providecommand{\doi}[1]{doi: #1}\else
  \providecommand{\doi}{doi: \begingroup \urlstyle{rm}\Url}\fi

\bibitem[Achenbach(1975)]{Achen75}
J.~Achenbach.
\newblock \emph{Wave Propagation in Elastic Solids}.
\newblock North-Holland Series in Applied Mathematics and Mechanics. Elsevier,
  Amsterdam, 1975.

\bibitem[Benzoni-Gavage et~al.(2002)Benzoni-Gavage, Rousset, Serre, and
  Zumbrun]{BRSZ}
S.~Benzoni-Gavage, F.~Rousset, D.~Serre, and K.~Zumbrun.
\newblock Generic types and transitions in hyperbolic initial-boundary value
  problems.
\newblock \emph{Proceedings of the Royal Society of Edinburgh Section A:
  Mathematics}, 132\penalty0 (5):\penalty0 1073--1104, 2002.
\newblock URL \url{https://doi.org/10.1017/S030821050000202X}.

\bibitem[Giang and Vinh(0)]{Pham24}
P.~T.~H. Giang and P.~C. Vinh.
\newblock {On the existence of Rayleigh waves with full impedance boundary
  condition}.
\newblock \emph{Mathematics and Mechanics of Solids}, 0\penalty0 (0):\penalty0
  10812865241266809, 0.
\newblock URL \url{https://doi.org/10.1177/10812865241266809}.

\bibitem[Godoy et~al.(2012)Godoy, Dur\'an, and N\'ed\'elec]{Godoy1}
E.~Godoy, M.~Dur\'an, and J.-C. N\'ed\'elec.
\newblock On the existence of surface waves in an elastic half-space with
  impedance boundary conditions.
\newblock \emph{Wave Motion}, 49\penalty0 (6):\penalty0 585--594, 2012.
\newblock URL \url{https://doi.org/10.1016/j.wavemoti.2012.03.005}.

\bibitem[Hayes and Rivlin(1962)]{HaRiv62}
M.~Hayes and R.~Rivlin.
\newblock A note on the secular equation for {R}ayleigh waves.
\newblock \emph{Zeitschrift Fur Angewandte Mathematik Und Physik - ZAMP},
  13:\penalty0 80--83, 07 1962.
\newblock URL \url{http://dx.doi.org/10.1007/BF01600759}.

\bibitem[Kaur and Singh(2021)]{KaurSingh21}
B.~Kaur and B.~Singh.
\newblock Rayleigh waves on the impedance boundary of a rotating monoclinic
  half-space.
\newblock \emph{Acta Mechanica}, 232\penalty0 (6):\penalty0 2479--2491, 2021.
\newblock URL \url{https://doi.org/10.1007/s00707-021-02959-w}.

\bibitem[Li(2006)]{Li2006}
X.-F. Li.
\newblock {On approximate analytic expressions for the velocity of Rayleigh
  waves}.
\newblock \emph{Wave Motion}, 44\penalty0 (2):\penalty0 120--127, 2006.
\newblock URL \url{https://doi.org/10.1016/j.wavemoti.2006.07.003}.

\bibitem[Malischewsky(1987)]{Masky1}
P.~Malischewsky.
\newblock \emph{Surface waves and discontinuities}.
\newblock 1 1987.
\newblock URL \url{https://www.osti.gov/biblio/6948756}.

\bibitem[Malischewsky(2000)]{Masky2000}
P.~G. Malischewsky.
\newblock {Comment to ``A new formula for the velocity of Rayleigh waves'' by
  D. Nkemzi [Wave Motion 26 (1997) 199--205]}.
\newblock \emph{Wave Motion}, 31\penalty0 (1):\penalty0 93--96, 2000.
\newblock URL \url{https://doi.org/10.1016/S0165-2125(99)00025-6}.

\bibitem[Malischewsky~Auning(2004)]{Masky04}
P.~G. Malischewsky~Auning.
\newblock A note on {R}ayleigh-wave velocities as a function of the material
  parameters.
\newblock \emph{Geof\'isica Internacional}, 2004.
\newblock URL \url{https://www.redalyc.org/articulo.oa?id=56843314}.

\bibitem[Pham and Vinh(2021)]{Pham21}
H.~G. Pham and P.~Vinh.
\newblock Existence and uniqueness of {R}ayleigh waves with normal impedance
  boundary conditions and formula for the wave velocity.
\newblock \emph{Journal of Engineering Mathematics}, 130:\penalty0 13, 10 2021.
\newblock URL \url{https://doi.org/10.1007/s10665-021-10170-y}.

\bibitem[Rahman and Barber(1995)]{Rahman1995}
M.~Rahman and J.~R. Barber.
\newblock {Exact Expressions for the Roots of the Secular Equation for Rayleigh
  Waves}.
\newblock \emph{Journal of Applied Mechanics}, 62\penalty0 (1):\penalty0
  250--252, 03 1995.
\newblock URL \url{https://doi.org/10.1115/1.2895917}.

\bibitem[Rahman and Michelitsch(2006)]{Rahman06}
M.~Rahman and T.~Michelitsch.
\newblock A note on the formula for the {R}ayleigh wave speed.
\newblock \emph{Wave Motion}, 43\penalty0 (3):\penalty0 272--276, 2006.
\newblock URL \url{https://doi.org/10.1016/j.wavemoti.2005.10.002}.

\bibitem[Ru(2023)]{CQRu23}
C.~Ru.
\newblock Rayleigh waves in an elastic half-space with a hard sphere-filled
  metasurface.
\newblock \emph{Mechanics Research Communications}, 131:\penalty0 104148, 2023.
\newblock URL \url{https://doi.org/10.1016/j.mechrescom.2023.104148}.

\bibitem[Singh and Kaur(2021)]{KaurSingh21a}
B.~Singh and B.~Kaur.
\newblock Rayleigh surface wave at an impedance boundary of an incompressible
  micropolar solid half-space.
\newblock \emph{Mechanics of Advanced Materials and Structures}, 29:\penalty0
  3942 -- 3949, 2021.
\newblock URL \url{https://doi.org/10.1080/15376494.2021.1914795}.

\bibitem[Tiersten(1969)]{Tier1}
H.~F. Tiersten.
\newblock Elastic surface waves guided by thin films.
\newblock \emph{Journal of Applied Physics}, 40\penalty0 (2):\penalty0
  770--789, 1969.
\newblock URL \url{https://doi.org/10.1063/1.1657463}.

\bibitem[Ting(2002)]{Ting2002}
T.~C.~T. Ting.
\newblock {An Explicit Secular Equation for Surface Waves in an Elastic
  Material of General Anisotropy}.
\newblock \emph{The Quarterly Journal of Mechanics and Applied Mathematics},
  55\penalty0 (2):\penalty0 297--311, 05 2002.
\newblock URL \url{https://doi.org/10.1093/qjmam/55.2.297}.

\bibitem[Vallejo(2025)]{Fab23}
F.~Vallejo.
\newblock The secular equation for elastic surface waves under boundary
  conditions of impedance type: A perspective from linear algebra.
\newblock \emph{Wave Motion}, 134:\penalty0 103476, 2025.
\newblock URL \url{https://doi.org/10.1016/j.wavemoti.2024.103476}.

\bibitem[Vinh and Ogden(2004)]{VinhOg2004}
P.~C. Vinh and R.~Ogden.
\newblock On formulas for the {R}ayleigh wave speed.
\newblock \emph{Wave Motion}, 39\penalty0 (3):\penalty0 191--197, 2004.
\newblock URL \url{https://doi.org/10.1016/j.wavemoti.2003.08.004}.

\bibitem[Vinh and {Thanh Hue}(2014)]{VinhThi14}
P.~C. Vinh and T.~T. {Thanh Hue}.
\newblock Rayleigh waves with impedance boundary conditions in anisotropic
  solids.
\newblock \emph{Wave Motion}, 51\penalty0 (7):\penalty0 1082--1092, 2014.
\newblock URL \url{https://doi.org/10.1016/j.wavemoti.2014.05.002}.

\bibitem[Vinh and Xuan(2017)]{Vinh17}
P.~C. Vinh and N.~Q. Xuan.
\newblock {Rayleigh waves with impedance boundary condition: Formula for the
  velocity, existence and uniqueness}.
\newblock \emph{European Journal of Mechanics - A/Solids}, 61:\penalty0
  180--185, 2017.
\newblock URL \url{https://doi.org/10.1016/j.euromechsol.2016.09.011}.

\end{thebibliography}

\end{document}